\documentclass[11pt]{article}
\usepackage{fullpage}
\usepackage{graphicx}
\usepackage{bbm}
\usepackage{amsmath,amssymb, amsfonts, xspace, amsthm}
\usepackage{color}
\usepackage{natbib}
\newtheorem{Theorem}{Theorem}
\newtheorem{Theorem*}{Theorem}

\newtheorem{Claim*}[Theorem]{Claim}

\newtheorem{CounterExample*}{$\overline{\hbox{\bf Example}}$}

\newtheorem{Example*}[Theorem]{Example}

\newtheorem{Intuition*}[Theorem]{Intuition}
\newtheorem{Joke*}[Theorem]{Joke}
\newtheorem{Lemma}[Theorem]{Lemma}
\newtheorem{Lemma*}[Theorem]{Lemma}
\newtheorem{Open problem}[Theorem]{Open problem}

\newtheorem{Question*}[Theorem]{Question}

\def \bSubexa    {\begin{subexa}}

\newcommand{\ignore}[1]{}

\newcommand{\II}{\mathbb{I}} % Added by Theertha on April 16th 2013.
\newcommand{\EE}{\mathbb{E}}

\newcommand{\NN}{\mathbb{N}}

\newcommand{\RR}{\mathbb{R}}

\newcommand{\naturals}{\NN}

\newcommand{\reals}{\RR}

\def \cE     {{\cal E}}

\def \cO     {{\cal O}}
\def \cP     {{\cal P}}

\def \cX     {{\cal X}}

\newcommand{\poi}{{\rm poi}}

\newcommand{\iid}{\textit{i.i.d.}} % Edit by Theertha, \xspace not compiling

\newcommand{\mrgs}[2]{\# #1 \#{\marginpar{#2}}}

\newcommand{\mrgchk}[1]{\mrgs{#1}{CHECK}}

\def \upto  {{,}\ldots{,}}

\def \sets#1{{\{#1\}}}
\def \Sets#1{{\left\{#1\right\}}}

\def \set#1#2{{\sets{{#1}\upto{#2}}}}

\newcommand{\ed}{\stackrel{\mathrm{def}}{=}}

\def\ignore#1{}

\newcommand{\bi}{\begin{itemize}}
\newcommand{\ei}{\end{itemize}}
\def\orpro{\mathop{\mathchoice
   {\vee\kern-.49em\raise.7ex\hbox{$\cdot$}\kern.4em}
   {\vee\kern-.45em\raise.63ex\hbox{$\cdot$}\kern.2em}
   {\vee\kern-.4em\raise.3ex\hbox{$\cdot$}\kern.1em}
   {\vee\kern-.35em\raise2.2ex\hbox{$\cdot$}\kern.1em}}\limits}

\def\andpro{\mathop{\mathchoice
 {\wedge\kern-.46em\lower.69ex\hbox{$\cdot$}\kern.3em}
 {\wedge\kern-.46em\lower.58ex\hbox{$\cdot$}\kern.25em}
 {\wedge\kern-.38em\lower.5ex\hbox{$\cdot$}\kern.1em}
 {\wedge\kern-.3em\lower.5ex\hbox{$\cdot$}\kern.1em}}\limits}

\def\simge{\mathrel{%
   \rlap{\raise 0.511ex \hbox{$>$}}{\lower 0.511ex \hbox{$\sim$}}}}

\def\simle{\mathrel{
   \rlap{\raise 0.511ex \hbox{$<$}}{\lower 0.511ex \hbox{$\sim$}}}}

\newcommand{\Dk}{\Delta_k}
\newcommand{\Dpsi}{\Delta_\psi}
\newcommand{\Mk}{M_k}
\newcommand{\Rw}{\hat R}
\newcommand{\RwMkn}{\hat R(\Mk^n)}
\newcommand{\zipf}[2]{\text{zipf}(#1,#2)}

\newcommand{\uenv}{\cE_{(f)}}
\newcommand{\uialphaen}{\cE_{(ci^{-\alpha},k)}}

\newcommand{\oenv}{\cE_{f}}
\newcommand{\udist}{\cP_{(p)}}
\newcommand{\uialpha}{\cP_{(\text{zipf}(\alpha,k))}}
\newcommand{\Ckalpha}{C_{k,\alpha}}
\newcommand{\pd}{\cP_d}
\newcommand{\ptau}{\tau^k}
\newcommand{\nz}{\Dpoin}
\newcommand{\dpoi}{d^{poi(n)}}
\newcommand{\Dpoin}{\varphi_+^{poi(n)}}
\newcommand{\Dn}{\varphi_+^n}
\newcommand\numberthis{\addtocounter{equation}{1}\tag{\theequation}}

\usepackage[backref,colorlinks,citecolor=blue,bookmarks=true]{hyperref}
\usepackage{aliascnt}
\usepackage[numbered]{bookmark}

\title{Universal Compression of Power-Law Distributions}
\author{
\begin{tabular}[t]{c@{\extracolsep{5em}}c@{\extracolsep{5em}}c} 
Moein Falahatgar  & Ashkan Jafarpour & Alon Orlitsky \\
\small\texttt{moein@ucsd.edu} & \small\texttt{ashkan@ucsd.edu} & \small\texttt{alon@ucsd.edu}
\end{tabular}
\vspace{2ex} \\
\begin{tabular}[t]{c@{\extracolsep{8em}}c} 
Venkatadheeraj Pichapati  & Ananda Theertha Suresh\\
\small\texttt{dheerajpv7@gmail.com} & \small\texttt{asuresh@ucsd.edu} 
\end{tabular}
\vspace{2ex}\\
University of California, San Diego
}
\begin{document}

\maketitle

\begin{abstract}
English words and the outputs of many other natural processes
are well-known to follow a Zipf distribution.
Yet this thoroughly-established property has never been shown to
help compress or predict these important processes.
We show that the expected redundancy of Zipf distributions of
order $\alpha>1$ is roughly the $1/\alpha$ power of the
expected redundancy of unrestricted distributions. Hence for these
orders, Zipf distributions can be better compressed and predicted
than was previously known.
Unlike the expected case, we show that worst-case
redundancy is roughly the same for Zipf and for unrestricted
distributions. Hence Zipf distributions have significantly
different worst-case and expected redundancies, making them
the first natural distribution class shown to have such a difference.
\end{abstract}

\textbf{Keywords: Power-law, Zipf, Universal Compression, Distinct elements, Redundancy} 

\section{Introduction}
\subsection{Definitions}
The fundamental
data-compression theorem states that every discrete
distribution $p$ 
can be compressed to its entropy
$H(p)\ed\sum p(x)\log\frac1{p(x)}$, a compression
rate approachable by assigning each symbol $x$ a codeword of roughly
$\log\frac1{p(x)}$ bits.

In reality, the underlying distribution is seldom known.
For example, in text compression, we observe only the words,
no one tells us their probabilities. In all these cases, it is
not clear how to compress the distributions to their entropy.

The common approach to these cases is \emph{universal
compression}. It assumes that while the underlying distribution is unknown,
it belongs to a known class of possible distributions,
for example, \iid\ or Markov distributions. Its goal is to derive
an encoding that works well for all distributions in the class.

To move towards formalizing this notion, observe that every
compression scheme for a distribution over a discrete set $\cX$
corresponds to some distribution $q$ over $\cX$
where each symbol $x\in\cX$ is assigned a codeword of length
$\log\frac{1}{q(x)}$. Hence the expected number of
bits used to encode the distribution's output is 
$\sum p(x)\log\frac{1}{q(x)}$, and the additional number of bits
over the entropy minimum is $\sum p(x)\log\frac{p(x)}{q(x)}$.

Let $\cP$ be a collection of distributions over $\cX$.
The collection's \emph{expected redundancy}, is the least worst-case
increase in the expected number of bits over the entropy, where the 
worst case is taken over all distributions in $\cP$ and the least 
is minimized over all possible encoders, 
\[
\bar{R}(\cP)
\ed
\min_q\max_{p\in\cP}\sum_{x\in\cX} p(x) \log \frac{p(x)}{q(x)}.
\]

An even stricter measure of the increased encoding length due to not
knowing the distribution is the collection's \emph{worst-case redundancy}
that considers the worst increase not just over all
distributions, but also over all possible outcomes $x$,
\[
\hat{R}(\cP)
\ed
\min_q\max_{p \in\cP} \max_{x \in \mathcal{X}}   \log \frac{p(x)}{q(x)}.
\]
Clearly,
\[
\bar{R}(\mathcal{P})
\le
\hat{R}(\mathcal{P}).
\]
Interestingly, until now, except for some made-up examples,
all analyzed collections had extremely close expected and worst-case
redundancies. One of our contributions is to demonstrate a practical
collection where these redundancies vastly differ, hence achieving 
different optimization goals may require different encoding
schemes.

%To show sharper results,
%we prove upper bounds on $\hat{R}$ and lower bounds on $\bar{R}$.

By far the most widely studied are the collections of \iid\ distributions. For every distribution $p$, the \iid\ 
distribution $p^n$ assigns to a length-$n$ string $x^n\ed (x_1,x_2\upto
x_n)$ probability $p(x^n)=p(x_1)\cdot\ldots\cdot p(x_n)$.
For any collection $\cP$ of distributions, the length-$n$
\iid\ collection is
\[
\cP^n
\ed
\Sets{p^n:p\in\cP}.
\]

\subsection{Previous results}
Let $\Dk$ denote the collection of all distribution over $\set1k$, 
where $\Delta$ was chosen to represent the simplex.
For the first few decades of universal compression,
researchers studied the redundancy of $\Dk^n$ when the alphabet size $k$ is fixed and the block length $n$ tends to infinity.
A sequence of papers~\cite{KT81,Kie78,Dav73,DMPW81,WST95,XB00,SW10,OS03:soi,Ris96,Cov91,Szp98,szpankowski2012minimax} showed that 
\[
\hat{R}(\Dk^n) = \frac{k-1}{2} \log \frac{n}{2 \pi} +\log \frac{\Gamma(\frac{1}{2})^k}{\Gamma(\frac{k}{2})}+ o_k(1),
\]
and that the expected redundancy is extremely close, at most
$\log e$ bits lower.
Note that a similar result holds for the complementary regime where
$n$ is fixed and $k$ tends to infinity,
\[
\hat{R}(\Dk^n) = n \log \frac{k}{n} + o(n).
\]
These positive results show that redundancy grows only
logarithmically with the sequence length $n$, therefore for long
sequences, the per-symbol redundancy diminishes to zero and the 
underlying distribution needs not to be known to approach entropy.
As is also well known, expected redundancy is exactly the same
as the log loss of sequential prediction, hence these results
also show that prediction can be performed with very small log loss.

However, as intuition suggests, and these equations confirm, redundancy
increases sharply with the alphabet size $k$.
In many, and possibly most, important real-life applications, 
the alphabet size is very large, often even larger than
the block length.
This is the case for example in applications involving natural
language processing, population estimation, and genetics~\cite{CG96}.
The redundancy in these cases is therefore very large, and can be
even unbounded for any sequence length $n$.

Over the last decade, researchers therefore considered methods
that could cope with compression and prediction of distributions
over large alphabets. Two main approaches were taken.

\cite{OSZ03} separated compression (and similarly prediction) of
large-alphabet sequences into compression of their 
\emph{pattern} that indicates the order at which symbols appeared,
and \emph{dictionary} that maps the order to the symbols.
For example, the pattern of $``banana"$ is $123232$ and its dictionary
is $1\to b$, $2\to a$, and $3\to n$. Letting $\Delta_\psi^n$
denote the collection of all pattern distributions, induced on sequences of length $n$
by all \iid\ distributions over any alphabet, a sequence of
papers~\cite{OSZ03,Sha06,Sha04,G09,OS03:soi,ADO12,acharya2013tight} showed that although patterns carry essentially all the entropy,
they can be compressed with redundancy
\[
0.3 \cdot n^{1/3}
\le
\bar R(\Dpsi^n)
\le
\hat R(\Dpsi^n)
\le
n^{1/3}\cdot\log^4 n
\] 
as $n \to \infty$. Namely, pattern redundancy too is sublinear in the block length and
most significantly, is uniformly upper bounded regardless of the alphabet
 size (which can be even infinite). It follows the per-symbol
pattern redundancy and prediction loss both diminish to zero at a
uniformly-bounded rate, regardless of the alphabet size.
Note also, that for pattern redundancy, worst-case and expected
redundancy are quite close.

However, while for many prediction applications predicting
the pattern suffices, for compression one typically needs to know
the dictionary as well. These results show that essentially all
the redundancy lies in the dictionary compression.

The second approach restricted the class of distributions
compressed. A series of works studied class of \emph{monotone} distributions~\cite{Sha13,Acharya14mono}. Recently, \cite{Acharya14mono} showed that the class $\Mk$ of
monotone distributions over $\set1k$ has redundancy
$\RwMkn \leq \sqrt{20n\log k \log n}$.

More closely related to this paper are \emph{envelope classes}.
An \emph{envelope} is a function $f:\naturals_+\to\reals_{\ge0}$.
For envelope function $f$, 
\[
\cE_f
\ed
\sets{p:p_i\le f(i)\text{ for all }i\ge1}
\]
is the collection of distributions where each $p_i$ is at
most the corresponding envelope bound $f(i)$.
Some canonical examples are the power-law envelopes 
$f(i)=c\cdot i^{-\alpha}$, and the exponential envelopes 
$f(i)=c\cdot e^{-\alpha\cdot i}$.
In particular, for power-law envelopes \cite{BGG09,BMesrob14} showed
\[
\Rw(\cE_f) \leq \left( \frac{2cn}{\alpha-1}\right)^{\frac{1}{\alpha}}(\log n)^{1-\frac{1}{\alpha}}+\cO(1),
\]
and more recently, \cite{acharya2014universal} showed that
\[
\Rw(\cE_f)= \Theta(n^{1/\alpha}).
\]

The restricted-distribution approach has the advantage that
it considers the complete
sequence redundancy, not just the pattern.
Yet it has the shortcoming that it may not capture 
relevant distribution collections.
For example, most real distributions are not monotone, 
words starting with `a' are not necessarily more likely
than those starting with `b'. 
Similarly for say power-law envelopes, why should words
in the early alphabet have higher upper bound than 
subsequent ones? Thus, words do not carry frequency order inherently. 

\subsection{Distribution model}
In this paper we combine the advantages and avoid the shortfalls
of both approaches to compress and predict distributions
over large alphabets.
As in patterns, we consider useful distribution collections, and
like restricted-distributions, we address the full redundancy.

Envelope distributions are very appealing as they effectively
represent our belief about the distribution. However their main drawback
is that they assume that the correspondence between the probabilities 
and symbols is known, namely that $p_i\le f(i)$ for the same $i$.
We relax this requirement and assume only that an upper envelope
on the sorted distribution, not the individual elements, is known.
Such assumptions on the sorted distributions are believed to hold
for a wide range of common distributions.

In 1935, linguist George Kingsley Zipf observed that when English words
are sorted according to their probabilities, namely so that
$p_1\ge p_2\ge\ldots$, the resulting distribution
follows a power law, $p_i\sim\frac{c}{i^{\alpha}}$
for some constant $c$ and power $\alpha$. \ignore{(see Figure~\ref{fig:word}).}
% is close to 1 and 
%$c$ is a normalization factor close to 0.1~\cite{zipf1935psycho}.
Long before Zipf,~\cite{pareto1896} studied distributions in income ranking and showed it can be mathematically expressed as power-law. Since then, researchers have found a very large number of distributions
such as word frequency, population ranks of cities, corporation sizes, and website users that when sorted follow this
\emph{Zipf-}, or \emph{power}-law ~\cite{zipf'32, zipf'49, adamic2002zipf}. 
%Several explanations for this phenomena have been proposed~\cite{}.
In fact, a Google Scholar search for ``power-law distribution''
returns around 50,000 citations. 
\ignore{
\begin{figure}[!t]
\includegraphics[scale=0.6]{wordGraph6}
\caption{Plot of log of the word frequency vs log of the rank}
\label{fig:word}
\end{figure}}
\ignore{
\begin{figure}
        \centering
        \begin{subfigure}[b]{0.2\textwidth}
                \caption{number of website users fitted by a power-law distribution with $\alpha=2.07$}.~\cite{}
                \label{fig:web_users}
        \end{subfigure}%
        ~ %add desired spacing between images, e. g. ~, \quad, \qquad, \hfill etc.
        \begin{subfigure}[b]{0.23\textwidth}
                \includegraphics[width=\textwidth]{wordGraph6.eps}
                \caption{word frequency vs rank of the words follows Zipf's law}
                \label{fig:word_freq}
        \end{subfigure}
        ~ %add desired spacing between images, e. g. ~, \quad, \qquad, \hfill etc.
        \caption{Examples of Zipf's law}\label{examples}
\end{figure}}

A natural question therefore is whether the established and commonly
trusted empirical observation that real distributions obey Zipf's law
can be used to better predict or equivalently compress them,
and if so, by how much.
\ignore{
Note that power-law and related distribution collections are a form 
of restricted-distributions. Take one, or several, fixed
distributions, and consider all their permutations.}
\ignore{
\mrgchk{Here you may want to introduce the notation and name you use
for the class of permutations of a distribution $p$ or permutations
of a collection $\cP$, for example the following. Maybe $\pi$ instead
of $\sigma$? Let me know if you want me to write it more fully.}
For any envelope $f$, let
\[
\cE_{\sigma(f)}
\ed 
\bigcup_{\sigma'} \cE_{\sigma'(f)}
\]
where $\sigma'$ ranges over all permutations of $\naturals_+$,
be the collection of distributions that are bounded by some
permutation of $f$.
Letting $\vec p$ and $\vec f$ denote the sorted 
version of $p$ and $f$, it is easy to see that
\mrgchk{Please verify}
\[
\cE_{\sigma(f)}
=
\Sets{p:\vec p\in\cE_{\vec f}},
\]
namely the set of distributions $p$ that when sorted are upper bounded
by the sorted version of $f$.
}

In Section~\ref{sec:prelim} we state our notation followed by new results in Section~\ref{sec:results}. Next, in Section~\ref{sec:worst} we bound the worst-case redundancy for power-law envelop class. In Section~\ref{sec:expected_distinct} we take a novel approach to analyze the expected redundancy. We introduce a new class of distributions which has the property that all permutations of a distribution are present in the class. Then we upper and lower bound the expected redundancy of this class based on the expected number of distinct elements. Finally, in Section~\ref{sec:zipf} we show that the redundancy of power-law envelop class can be studied in this framework.

\section{Preliminaries}
\label{sec:prelim}
\subsection{Notation}
Let $x^n\ed(x_1,x_2,..,x_n)$ denote a sequence of length $n$, 
 $\mathcal{X}$ be the underlying alphabet and $k \ed |\cX|$.  
 The \emph{multiplicity} $\mu_x$ of a symbol $x \in \cX$ is the number of times $x$ appears in $x^n$. Let $[k]=\{1,2,...,k\}$ be the indices of elements in $\cX$. The type vector of $x^n$ over $[k]=\{1,2,...,k\}$, $\tau(x^n)=(\mu_1,\mu_2,\ldots,\mu_k)$ is a $k$-tuple of multiplicities in $x^n$. The \emph{prevalence} of a multiplicity $\mu$, denoted by $\varphi_{\mu}$, is the number of elements appearing $\mu$ times in $x^n$. For example, $\varphi_1$ denotes the number of elements which appeared once in $x^n$. Furthermore, $\varphi_+$ denotes the number of distinct elements in $x^n$. The vector of prevalences for all $\mu$'s is called the profile vector.

We use $p_{(i)}$ to note the $i^{th}$ highest probability in $p$. Hence,
$p_{(1)} \geq p_{(2)} \geq \ldots p_{(k)}$.
Moreover, we use $\zipf{\alpha}{k}$ to denote Zipf distribution with parameter $\alpha$ and support $k$. Hence,
\[
\zipf{\alpha}{k}_i = \frac{i^{-\alpha}}{\Ckalpha},
\]
where $\Ckalpha$ is the normalization factor. Note that all logarithms in this paper are in base $2$ and we consider only the case $\alpha>1$.
\ignore{The profile $\varphi$ of a sequence is the multi-set of multiplicities of all symbols appearing.}

\subsection{Problem statement}

For an envelope $f$ with support size $k$, let $\uenv$ be the class of distributions such that
\[
\uenv = \{p :  \ p_{(i)} \leq f(i)\ \forall 1 \leq i \leq k \}.
\]
Note that $\oenv \subset \uenv$. 
\ignore{Our goal is to find the redundancy of envelope classes as a function of $f$ and use it for the power-law distributions.}
We also consider the special case when $f$ is a distribution itself, in which case we denote $\uenv$ by  $\udist$, a class that has distributions
whose multi-set of probabilities is same as $p$. In other words, $\udist$ contains all permutations of distribution $p$. Also we define 
\begin{align*}
\pd^n=\{p^n:  \EE_p[\Dn] \le d \},
\end{align*} 
where $\Dn$ is the number of distinct elements in $x^n$. Note that for any distribution belonging to this class, all permutations of it are also in the class. 

\section{Results}
\label{sec:results}
We first consider worst-case redundancy,
lower-bound it for general unordered permutations,
and apply the result to unordered power-law classes,
showing that for % $\alpha >1$ and 
$n \leq k^{1/\alpha}$,
\[
\hat{R}(\uialphaen^n) \geq \hat{R}(\uialpha^n) \geq n \log \frac{k-n}{n^{\alpha} \Ckalpha}.
\]
This shows that the worst-case redundancy of power-law
distributions behaves roughly as that of general distributions over the same alphabet.

More interestingly, we establish a general method for upper- and lower-bounding
the expected redundancy of unordered envelope distributions
in terms of expected number of distinct symbols. Precisely, for a class $\pd^n$ we show the following upper bound
\[
\bar{R}(\pd^n) \leq d \log \frac{kn}{d^2} + (2\log e+1)d + \log (n+1).
\]
\textbf{Interpretation}:
This upper bound can be also written as 
\begin{equation}
\label{eq:upperbound}
 \log n + \log {k \choose d} + \log {n-1 \choose d-1}. 
\end{equation}
This suggests a very clear intuition of the upper bound. We can give a compression scheme for any sequence that we observe. Upon observing a sequence $x^n$, first we declare how many distinct elements are in that sequence. For this we need $\log n$ bits. In addition to those bits, we need $ \log {k \choose d}$ bits to specify which $d$ distinct elements out of $k$ elements appeared in the sequence. Finally, for the exact number of occurrences of each distinct element we should use $\log {n-1 \choose d-1}$ bits.

We also show a lower bound which is dependent on both the expected number of distinct elements $d$ and the distributions in the class $\pd^n$. Namely, we show

\begin{align*}
\bar{R}(\pd^n) &\geq \left(\log {k\choose d}- d \log \frac{n}{d}-d \log \pi e\right)(1+o_d(1)) -\sum_{np_i < 0.7} (3np_i -np_i \log np_i) .
\end{align*}

Using this result, we then consider expected redundancy of power-law distributions as a special case of $\pd^n$
and show that it is significantly lower than that of general
distributions. This shows that on average, Zipf distributions
can be compressed much better than general ones.
Since expected redundancy is the same as log loss, 
they can also be predicted more effectively. In fact we show that for $k>n$,
\[
 \bar{R}(\uialphaen^n)= \Theta (n^{\frac{1}{\alpha}} \log k ).
\]

Recall that general length-$n$ \iid\ distributions over 
alphabet of size $k$ have redundancy roughly $n\log\frac kn$ bits.
Hence, when $k$ is not much larger than $n$, the expected redundancy
of Zipf distributions of order $\alpha>1$ is the $1/\alpha$ power
of the expected redundancy of general distributions. 
For example, for $\alpha=2$ and $k=n$, the redundancy of Zipf
distributions is $\Theta(\sqrt n\log n)$ compared to $n$ for
general distributions.
This reduction from linear to sub-linear dependence on $n$ also 
implies that unordered power-law envelopes are universally
compressible when $k=n$.

These results also show that worst-case redundancy is roughly the
same for Zipf and general distributions.
Comparing the results for worst-case and expected redundancy of 
Zipf distributions, it also follows
that for those distributions expected- and worst-case redundancy
differ greatly. This is the first natural class of distribution 
for which worst-case and expected redundancy have been
shown to significantly diverge.

As stated in the introduction, for the power-law envelope $f$, \cite{acharya2014universal} showed that
\[
\Rw(\cE_f)= \Theta(n^{1/\alpha}).
\]

Comparing this with the results in this paper reveals that if we know the envelop on the class of distributions but we do not know the true order of that, we have an extra multiplicative factor of $\log k$ in the expected redundancy, i.e. 
\[
\bar{R}(\cE_{(f)})= \Theta(n^{1/\alpha} \log k).
\]
\section{Worst-case redundancy}
\label{sec:worst}
\subsection{Shtarkov Sum}
It is well known that the worst-case redundancy can be calculated using Shtarkov sum~\cite{shtar1987universal}, i.e. for any class $\cP$
\begin{equation}
\label{redundancy-shtarkov}
\hat{R}(\mathcal{P})=\log S(\mathcal{P}),
\end {equation}
 where $S(\mathcal{P})$ is the Shtarkov sum and defined as
 \begin{equation}
\label{shtarkov}
S(\mathcal{P}) \ed \sum_{x \in \mathcal{X}} \hat{p}(x).
\end {equation}
For notational convenience we denote $\hat{p}(x) \ed \max_{p \in \cP}p(x)$,
to be the maximum probability any distribution in $\cP$ assigns to $x$.

\subsection{Small alphabet case}
Recall that $\hat{R}(\Dk^n) \approx \frac{k-1}{2} \log n$. 
We now give a simple example to show that unordered distribution classes $\cP_{(p)}$ may have much smaller redundancy. In particular we show that for a distribution $p$ over $k$ symbols, 
\[
\hat{R}(\udist^n) \leq  \log k! \leq k \log k \ \ \ \forall n.
\]
 Consider the Shtarkov sum
\begin{align*}
S(\udist^n) 
&= \sum_{x^n \in \cX^n} \hat{p}(x^n)\\
& \leq \sum_{x^n \in \cX^n} \sum_{p \in \udist} p(x^n) \\
& = \sum_{p \in\udist} \sum_{x^n \in \cX^n} p(x^n) \\
& = \sum_{p \in \udist} 1= |\udist| = k!.
\end{align*}
Clearly for $n \gg k$, the above bound is smaller than $\hat{R}(\Dk^n)$. 

\subsection{Large alphabet regime}

From the above result, it is clear that as $n \to \infty$, 
the knowledge of the underlying-distribution multi-set helps in universal compression. A natural question is to ask if the same applies for the large alphabet regime when the number of samples $n \ll k$.
Recall that \cite{acharya2014universal, BGG09} showed that for power-law envelopes, $f(i)=c\cdot i^{-\alpha}$, with infinite support size
\[
\hat{R}(\oenv) = \Theta(n^{\frac{1}{\alpha}}).
\]
We show that if the permutation of the distribution is not known
then the worst-case redundancy is $\Omega(n) \gg\Theta(n^{\frac{1}{\alpha}})$,
and thus the knowledge of the permutation
is essential.
In particular, we prove that
even for the case when the envelope class consists of only one power-law distribution, $\hat{R}$ scales as $n$.
\begin{Theorem}
\label{thm:worst}
For % $\alpha >1$ and 
$n \leq k^{1/\alpha}$,
\[
\hat{R}(\uialphaen^n) \geq \hat{R}(\uialpha^n) \geq n \log \frac{k-n}{n^{\alpha} \Ckalpha}.
\]
\end{Theorem}
\begin{proof}
Since $\uialpha^n \subset \uialphaen^n$, we have 
\[
\hat{R}(\uialphaen^n) \geq \hat{R}(\uialpha^n).
\]
To lower bound $ \hat{R}(\uialpha^n)$,
recall that 
\begin{align*}
S(\uialpha^n) 
&= \sum_{x^n} \hat{p} (x^n) \\
& \geq \sum_{x^n : \Dn = n} \hat{p}(x^n),
\end{align*}
where $\Dn$ is the number of distinct symbols in $x^n$.
Note that number of such sequences is $k(k-1)(k-2)\ldots(k-n+1)$.
We lower bound $\hat{p}(x^n)$ for every such sequence.
Consider the distribution $q \in \uialpha$
given by $q(x_i) = \frac{1}{i^{\alpha} \Ckalpha}\ \forall 1 \leq i \leq n$.
Clearly $\hat{p}(x^n) \geq q(x^n)$ and as a result we have
\begin{align*}
S(\uialpha^n) 
& \geq k(k-1)(k-2)\ldots (k-n+1) \prod^n_{i=1} \frac{1}{i^{\alpha} \Ckalpha} \\
& \geq \left(\frac{k-n}{n^{\alpha}\Ckalpha}\right)^n.
\end{align*}
Taking the logarithm yields the result.
\end{proof}
Thus for small values of $n$, independent of the underlying distribution per-symbol redundancy is $\log \frac{k}{n^{\alpha}}$.  
Since for $n \leq k$, $\hat{R}(\Dk^n) \approx n \log \frac{k}{n}$, we have 
for $n \leq k^{1/\alpha}$
\[
\hat{R}(\uialphaen^n) \leq \hat{R}(\Dk^n) \leq \cO ( n \log \frac{k}{n}).
\]
 Therefore, together with Theorem~\ref{thm:worst}, we have for $n \leq k^{1/\alpha}$ 
\[
 \Omega( n \log \frac{k}{n^{\alpha}})\le \hat{R}(\uialphaen^n) \le \cO ( n \log \frac{k}{n}).\]

\section{Expected redundancy based on the number of distinct elements}
\label{sec:expected_distinct}

In order to find the redundancy of the unordered envelop classes, we follow a more systematic approach and define another structure on the underlying class of distributions. More precisely, we consider the class of all distributions in which we have an upper bound on the expected number of distinct elements we are going to observe. Lets define 
\begin{align*}
\pd^n=\{p^n:  \EE_p[\Dn] \le d \},
\end{align*}
where $\Dn$ is the number of distinct symbols in the sequence $x^n$. Note that for any distribution belonging to this class, all permutations of it are also in the class. We later show that envelop classes can be described in this way and the expected number of distinct elements characterizes the envelop classes; therefore we can bound the redundancy of them applying results in this section. 
\subsection{Upper bound}
The following lemma bounds the expected redundancy of a class in terms of $d$.
\begin{Lemma}
\label{lem:distinctupper}
For any class $\pd^n$, 
\[
\bar{R}(\pd^n) \leq d \log \frac{kn}{d^2} + (2\log e+1)d + \log (n+1).
\]
\end{Lemma}
\begin{proof}
We give an explicit coding scheme that achieves the above redundancy. For a sequence $x^n$ with multiplicities of symbols $\mu^k \ed \mu_1,\mu_2,\ldots, \mu_k$, 
let 
\[
q(x^n) = \frac{1}{N_{\Dn}} \cdot \prod^k_{j=1} \left( \frac{\mu_j}{n}\right)^{\mu_j}
\]
be the probability our compression scheme assigns to $x^n$
and $N_{\Dn}$ is the normalization factor given by
\[
N_{\Dn} =  n \cdot {k \choose \Dn} \cdot {{n-1\choose {\Dn-1}}} .
\]
Before proceeding, we show that $q$ is a valid coding scheme by showing that $\sum_{x^n \in \cX^n} q(x^n) \leq 1$. 
We divide the set of sequences as follows.
\[
\sum_{x^n \in \cX^n}  = \sum^n_{d'=1} \ \  \sum_{S\in \cX : |S| = d'} \ \ \ 
\sum_{\mu^k : \mu_i=0 \text{ iff } i \notin S}\ \ 
\sum_{x^n: 
\mu(x^n) = \mu^k} 
\]
Now we can re-write and bound $\sum_{x^n \in \cX^n} q(x^n)$ as the following.
\begin{align*}
\sum^n_{d'=1} \ \  \sum_{S\in \cX : |S| = d'} \ \ \ 
\sum_{\mu^k : \mu_i=0 \text{ iff } i \notin S}\ \ 
\sum_{x^n: 
\mu(x^n) = \mu^k} q(x^n) &\stackrel{(a)}{\le}\sum^n_{d'=1} \ \  \sum_{S\in \cX : |S| = d'} \ \ \ 
\sum_{\mu^k : \mu_i=0 \text{ iff } i \notin S} \frac{1}{N_{d'}}\\
& \stackrel{(b)}{=} \sum^n_{d'=1} \frac{{k \choose d'} \cdot {{n-1\choose {d'-1}}}}{N_{d'}}\\
& = \sum^n_{d'=1} \frac{1}{n} = 1.
\end{align*}
where $(a)$ holds since for a given $\mu^k$ , the maximum likelihood distribution for all sequences with same values of $\mu_1,\mu_2,\ldots \mu_k$ are same. Also $(b)$ follows from the fact that  the second summation ranges over ${k \choose d'}$ values and the third summation ranges over ${n-1 \choose {d'-1}}$ values.
Furthermore for any $p^n \in \pd^n$,
\begin{align*}
\log \frac{p(x^n)}{q(x^n)} 
& \leq \log N_{\Dn} + n \cdot \sum^k_{i=1} \frac{\mu_i}{n} \log \frac{p_i}{\mu_i/n} \leq \log N_{\Dn}.
\end{align*}
Taking expectation over both sides
\begin{align*}
& \bar{R}(\uenv) \le \EE[\log N_{\Dn}] \\
& \leq \log n + \EE\left[\log {k \choose \Dn} + \log {n-1 \choose \Dn-1}\right] \\
& \stackrel{(a)}{\leq} \log n + \EE \left[\Dn \log \left( \frac{k}{\Dn} \cdot \frac{2n}{\Dn}\right) +(2 \log e)\Dn \right] \\
& \stackrel{(b)}{\leq} \log n + d \log \frac{kn}{d^2} + (2\log e+1)d,
\end{align*}
where $(a)$ follows from the fact that ${n \choose d} \leq \left(\frac{ n e}{d} \right)^d$ and $(b)$ follows from Jensen's inequality.
\end{proof}
\ignore{
\textbf{Interpretation}:
The upper bound given in Theorem~\ref{lem:distinctupper} can be also written as 
\begin{equation}
\label{eq:upperbound}
 \log n + \log {k \choose d} + \log {n-1 \choose d-1}. 
\end{equation}
This suggests a very clear intuition of the upper bound. We can give a compression scheme for any sequence that we observe. Upon observing a sequence $x^n$, first we declare how many distinct elements are in that sequence. For this we need $\log n$ bits. In addition to those bits, we need $ \log {k \choose d}$ bits to specify which $d$ distinct elements out of $k$ elements appeared in the sequence. Finally, for the exact number of occurrences of each distinct element we should use $\log {n-1 \choose d-1}$ bits, based on balls and bins model. Next we show that the leading term in~\eqref{eq:upperbound} is also the necessary number of extra bits to compress the sequence $x^n$. In other words, the upper bound and the lower bound on the redundancy are matched in leading terms.}

\subsection{Lower bound}
To show a lower bound on the expected redundancy of class $\pd^n$, we use some helpful results introduced in previous works. First, we introduce Poisson sampling and relate the expected redundancy in two cases when we use normal sampling and Poisson sampling. Then we prove the equivalence of expected redundancy of the sequences and expected redundancy of types. 

\textbf{Poisson sampling}: In the standard sampling method, where a distribution is sampled $n$ times, 
the multiplicities are dependent, for example they add up to $n$.
Hence, calculating redundancy under this sampling often requires 
various concentration inequalities, complicating the proofs. A useful approach to make them independent and hence simplify the analysis is to sample the distribution $n'$ times, where $n'$ is a Poisson random variable with mean $n$.
Often called as Poisson sampling, this approach has been used in universal compression to simplify the analysis~\cite{ADO12,acharya2013tight,yang2013large,acharya2014universal}. 

Under Poisson sampling, if a distribution $p$ is sampled $\iid$ $\poi(n)$ times, then the number of times symbol $x$ appears is an independent Poisson random variable with mean $np_x$,
namely, $\Pr(\mu_x = \mu) = \frac{e^{-np_x} (np_x)^{\mu}}{\mu!}$~\cite{MU05}. Henceforth, to distinguish between two cases of normal sampling and Poisson sampling we specify it with superscripts $n$ for normal sampling and $poi(n)$ for Poisson sampling. 

Next lemma lower bounds $\bar{R}(\cP^n)$ by the redundancy in the presence of Poisson sampling. We use this lemma further in our lower-bound arguments. 
\begin{Lemma}
\label{lem:lower_poisson}
For any class $\cP$,
\[
\bar{R}(\cP^n) \geq \frac{1}{2}\bar{R}(\cP^{\poi(n)}).
\]
\end{Lemma}
\begin{proof}
By the definition of $\bar{R}(\cP^{\poi(n)})$,
\begin{equation}
\label{eq:poi_ave}
\bar{R}(\cP^{\poi(n)}) = \min_{q} \max_{p \in \cP} \EE_{\poi(n)} \left[\log \frac{p_{\poi(n)}(x^{n'})}{q(x^{n'})} \right],
\end{equation}
where subscript $\poi(n)$ indicates that the probabilities are calculated under Poisson sampling. Similarly, for every $n'$,
\[
\bar{R}(\cP^{n'}) = \min_{q} \max_{p \in \cP} \EE\left[\log \frac{p(x^{n'})}{q(x^{n'})} \right].
\]
Let $q_{n'}$ denote the distribution that achieves the above minimum. We upper bound the right hand side of Equation~\eqref{eq:poi_ave} by constructing an explicit $q$. Let 
\[
q(x^{n'}) = e^{-n}\frac{n^{n'}}{n'!} q_{n'}(x^{n'}).
\]
Clearly $q$ is a distribution as it adds up to $1$. Furthermore, since $p_{\poi(n)}(x^{n'}) = e^{-n}\frac{n^{n'}}{n'!} p(x^{n'})$,
we get
\begin{align*}
\bar{R}(\cP^{\poi(n)})
& \leq  \max_{p \in \cP} \EE_{\poi(n)} \left[\log \frac{p_{\poi(n)}(x^{n'})}{q(x^{n'})} \right] \\
&  =  \max_{p \in \cP} \sum^{\infty}_{n'=0} 
e^{-n}\frac{n^{n'}}{n'!}
 \EE \left[\log \frac{e^{-n}\frac{n^{n'}}{n'!}  p(x^{n'})}{e^{-n}\frac{n^{n'}}{n'!}  q_{n'}(x^{n'})} \right] \\
& \leq \sum^{\infty}_{n'=0} 
e^{-n}\frac{n^{n'}}{n'!} \max_{p \in \cP} \EE \left[ \log \frac{p(x^{n'})}{ q_{n'}(x^{n'})} \right] \\
& = \sum^{\infty}_{n'=0} e^{-n}\frac{n^{n'}}{n'!}  \bar{R}(\cP^{n'}),
\end{align*}
where the last equality follows from definition of $q_{n'}$.
By monotonicity and sub-additivity of $\bar{R}(\cP^{n'})$ 
(see Lemma $5$ in~\cite{ADO12}), it follows that
\begin{align*}
  \bar{R}(\cP^{n'}) 
&\leq \bar{R}(\cP^{n\lceil \frac{n'}{n} \rceil}) \\
& \leq \left\lceil \frac{ n'}{n} \right\rceil \bar{R}(\cP^{n}) \\
& \leq  \left(\frac{n'}{n}+1\right) \bar{R}(\cP^{n}).
\end{align*}
Substituting the above bound we get
\begin{align*}
\bar{R}(\cP^{\poi(n)})
& \leq \sum^{\infty}_{n'=0} e^{-n}\frac{n^{n'}}{n'!}  \left(\frac{ n'}{n}+1\right) \bar{R}(\cP^{n}) \\
& = 2 \bar{R}(\cP^{n}),
\end{align*}
where the last equality follows from the fact that expectation of $n'$ is $n$.
\end{proof}

\textbf{Type redundancy}: In the following lemma we show that the redundancy of the sequence is same as the redundancy of the type vector. Therefore we can focus on compressing the type of the sequence and calculate the expected redundancy of that. 
\begin{Lemma}
\label{lem:typeredundancy}
Lets define $\tau(\cP^n)=\{\tau(p^n): p\in \cP\}$, then we have
\[
\bar{R}(\tau(\cP^n)) = \bar{R}(\cP^n).
\]
\end{Lemma}
\begin{proof}
\begin{align*}
\bar{R}(\cP^n)&= \min_q \max_{p \in \cP} \EE\left[\log \frac{p(x^n)}{q(x^n)}\right]\\
=&\min_q \max_{p \in \cP} \sum_{x^n \in \cX^n}p(x^n)\log \frac{p(x^n)}{q(x^n)}\\
=&\min_q \max_{p \in \cP} \sum_{\tau} \sum_{x^n \in \cX^n: \tau(x^n)=\tau}p(x^n)\log \frac{p(x^n)}{q(x^n)}\\
\stackrel{(a)}{=}& \min_q \max_{p \in \cP} \sum_{\tau} \left(\sum_{x^n: \tau(x^n)=\tau} p(x^n)\right) \log \frac{\sum_{x^n: \tau(x^n)=\tau} p(x^n)}{\sum_{x^n: \tau(x^n)=\tau} q(x^n)}\\
=& \min_q \max_{p \in \cP} \sum_{\tau} p(\tau)\log \frac{p(\tau)}{q(\tau)}\\
=& \bar{R}(\tau(\cP^n))
\end{align*}
where $(a)$ is by convexity of KL-divergence and the fact that all sequences of a specific type have the same probability. 
\end{proof}

Now we reach to the main part of this section, i.e. lower bounding the expected redundancy of class $\pd^n$. Based on the previous lemmas, we have 
\[
\bar{R}(\pd^n) \ge \frac12 \bar{R}(\pd^{poi(n)})=\frac12\bar{R}(\tau(\pd^{poi(n)}))
\]
and therefore it is enough to show a lower bound on $\bar{R}(\tau(\pd^{poi(n)}))$. We decompose $\bar{R}(\tau(\pd^{poi(n)}))$ as
\begin{align*}
\bar{R}(\tau(\pd^{poi(n)}))&=\min_q \max_{\ptau \in \tau(\pd^{poi(n)})} \sum_{\tau^k} p(\tau^k) \log \frac{p(\tau^k)}{q(\tau^k)}\\
&=\min_q \max_{\ptau \in \tau(\pd^{poi(n)})} \sum_{\tau^k} p(\tau^k) \log \frac{1}{q(\tau^k)}\\
&-\sum_{\tau^k} p(\tau^k) \log \frac{1}{p(\tau^k)}\
\end{align*}
Hence it suffices to show a lower bound on $\sum_{\tau^k} p(\tau^k) \log \frac{1}{q(\tau^k)}$ and an upper bound on $\sum_{\tau^k} p(\tau^k) \log \frac{1}{p(\tau^k)}$. For the first term, we upper bound $q(\tau^k)$ based on the number of distinct elements in sequence $x^{poi(n)}$. Lemmas~\ref{lem:dpoi1},~\ref{lem:dpoi2},~\ref{lem:boundond} prove this upper bound. Afterwards we consider the second term and it turns out that this term is nothing but the entropy of the type vectors under Poisson sampling. 

The following two concentration lemmas from ~\cite{GnedinHP2007,BenBO2014} help us to relate the expected number of distinct elements for normal and Poisson sampling. We continue by a lemma making connection between those two quantities. Denote the number of distinct elements in $x^{poi(n)}$ as $\Dpoin$, and $\dpoi=\EE[\Dpoin]$. Similarly, $\Dn$ is the number of distinct elements in $x^n$ and $d=\EE[\Dn]$.

\begin{Lemma}
\label{lem:dpoi1}  
(\cite{BenBO2014}) Let $v=\EE[\varphi_1^{poi(n)}]$ be the expected number of elements which appeared once in $x^{poi(n)}$, then 
\[
 \Pr[\Dpoin<\dpoi-\sqrt{2vs}] \le e^{-s}.
 \]
\end{Lemma}

\begin{Lemma}
\label{lem:dpoi2}
(Lemma 1 in~\cite{GnedinHP2007}) Let $\EE[\varphi_2^{poi(n)}]$ be the expected number of elements which appeared twice in $x^{poi(n)}$, then 
\[
|\dpoi-d|<2 \frac{\EE[\varphi_2^{poi(n)}]}{n}.
\]
\end{Lemma}
Using Lemmas~\ref{lem:dpoi1} and~\ref{lem:dpoi2} we lower and upper bound the number of non-zero elements in $\tau(x^{poi(n)})$.
\begin{Lemma}
\label{lem:boundond}
The number of non-zero elements in $\tau(x^{poi(n)})$ is more than $(1-\epsilon)d$ with probability $>1-e^{-\frac{d(\epsilon-2/n)^2}{2}}$. Also, the number of non-zero elements in $\tau(x^{poi(n)}) < (1+\epsilon)d$ with probability $>1-e^{-\frac{d(\epsilon-2/n)^2}{2}}$.
\end{Lemma}
\begin{proof}
The number of non-zero elements in $\tau$ is equal to the number of distinct elements in $x^{poi(n)}$. By Lemma~\ref{lem:dpoi1}
\begin{align*}
\Pr[\Dpoin<\dpoi(1-\epsilon)] &\le e^{-\frac{(\dpoi\epsilon)^2}{2v}}\\
& \stackrel{(a)}{\le} e^{-\frac{\dpoi \epsilon^2}{2}},
\end{align*}
where $(a)$ is because $\dpoi>v$. Lemma~\ref{lem:dpoi2} implies $ \dpoi(1-\frac2n)< d < \dpoi(1+\frac2n) $. Therefore, 
\begin{align*}
\Pr[\Dpoin < d(1-\epsilon)] &\le \Pr[\Dpoin<\dpoi\left(1+\frac2n\right)(1-\epsilon)] \\
&\le e^{-\frac{d (\epsilon-\frac2n)^2}{2}}. 
\end{align*}
Proof of the other part is similar and omitted.
 \end{proof}
Next, we lower bound the number of bits we need to express $\tau^k$ based on the number of nonzero elements in it. 
\begin{Lemma}
\label{lem:boundonq}
 If number of non-zero elements in $\tau^k$ is more than $d'$, then 
 \[
 q(\tau^k) \leq \frac{1}{{k\choose d'}}.
 \]
\end{Lemma}
\begin{proof}
Consider all the type vectors with the same number of non-zero elements as $\tau^k$. It is not hard to see that $q$ should assign same probability to all types with the same profile vector. 
 Number of such type vectors for a given number of non-zero elements $d'$ is at least $k \choose d'$.
\end{proof}
Note that the number of non-zero elements in $\tau^k$ is same as $\Dpoin$. Based on Lemmas~\ref{lem:boundond} and  \ref{lem:boundonq} we have
\begin{align}
\sum_{\tau^k} p(\tau^k) \log \frac{1}{q(\tau^k)} &\ge \sum_{\tau^k: \nz \ge (1-\epsilon)d} p(\tau^k) \log \frac{1}{q(\tau^k)}\nonumber\\
&\ge \sum_{\tau^k: \nz \ge (1-\epsilon)d} p(\tau^k) \log {k\choose d(1-\epsilon)}\nonumber\\
&\ge \left(1- e^{-\frac{d (\epsilon-\frac2n)^2}{2}}\right) \log {k\choose d(1-\epsilon)}\nonumber\\
&= \log {k\choose d}(1+o_d(1))\label{lbound1}.
\end{align}
where the last line is by choosing $\epsilon=d^{-\frac13}$. 
Now we focus on bounding the entropy of the type. Recall that if distribution $p$ is sampled $\iid$ $ poi(n)$ times, then the number of times symbol $i$ appears, $\mu_i$, is an independent Poisson random variable with mean $\lambda_i = np_i$. First we state a useful lemma in calculation of the entropy.
\begin{Lemma}
\label{lem:poientropy}
If $X \sim poi(\lambda)$ for $\lambda<1$, then
\[
H(X) \le \lambda[1-\log \lambda]+e^{-\lambda} \frac{\lambda^2}{1-\lambda}.
\]
\end{Lemma}
\begin{proof} 
\begin{align*}
H(X)&=-\sum_{i=0}^{\infty} p_i \log p_i \\
&=-\sum_{i=0}^{\infty}e^{-\lambda} \frac{\lambda^i}{i!} \log \frac{e^{-\lambda}\lambda^i}{i!}\\
&=-\sum_{i=0}^{\infty}e^{-\lambda} \frac{\lambda^i}{i!}\left[ \log e^{-\lambda}+i\log \lambda - \log (i!)\right]\\
&=\lambda \sum_{i=0}^{\infty}e^{-\lambda} \frac{\lambda^i}{i!}-\log \lambda\sum_{i=0}^{\infty} ie^{-\lambda} \frac{\lambda^i}{i!}+\sum_{i=0}^{\infty}e^{-\lambda} \frac{\lambda^i}{i!} \log(i!)\\
&\stackrel{(a)}{=}\left[\lambda - \lambda \log \lambda \right]+e^{-\lambda}\left[\sum_{i=2}^{\infty} \frac{\lambda^i \log (i!)}{i!}\right]\\
& \le \left[\lambda - \lambda \log \lambda \right]+e^{-\lambda} \left[\sum_{i=0}^{\infty} \lambda^i \right]\\
&\stackrel{(b)}{=}\lambda[1-\log \lambda]+e^{-\lambda} \frac{\lambda^2}{1-\lambda}
\end{align*}
where $(a)$ is because the first two terms in the last summation is zero and for the rest of the terms, $\log(i!)<i!$ Also $(b)$ follows from geometric sum for $\lambda<1$. 
\end{proof}
We can write
\begin{align*}
H(\tau^k)&=\sum_{i=1}^k H(\mu_i)\\
&= \sum_{i=1}^k H(poi(\lambda_i))\\
&=\sum_{\lambda_i < 0.7} H(poi(\lambda_i)) + \sum_{\lambda_i \ge 0.7} H(poi(\lambda_i))\\
&\stackrel{(a)}{=} \sum_{\lambda_i < 0.7} \left( \lambda_i - \lambda_i \log \lambda_i + e^{-\lambda_i}\frac{\lambda_i^2}{1-\lambda_i}\right)+ \sum_{\lambda_i \ge 0.7} H(poi(\lambda_i))\\
&\stackrel{(b)}{\le}   \sum_{\lambda_i < 0.7} \left( 3\lambda_i - \lambda_i \log \lambda_i \right)+ \sum_{\lambda_i \ge 0.7} \frac12 \log \left(2\pi e (\lambda_i +\frac{1}{12})\right)\numberthis \label{eq:entropy}
\end{align*}
where $(a)$ is due to  Lemma~\ref{lem:poientropy} and (b) is by using Equation $(1)$ in \cite{AS10} and the fact that $e^{-x}\frac{x^2}{1-x}<2x$  for $x<0.7$.

In the rest of this section, we calculate an upper bound for the second term in~\eqref{eq:entropy}.Note that the first term in the same equation, i.e. $\sum_{\lambda_i < 0.7} H(poi(\lambda_i))$ is heavily dependent on the shape of distributions in the class. In other words, upper bounding this term generally, will lead us to a weak lower bound, while plugging in the exact values leads to a matching lower bound for the intended envelope class, i.e. Zipf distributions.

Let $n^{-}$ be the sum of all $\lambda <0.7$ and $n^{+}$ be the sum of all $\lambda \ge 0.7$. Similarly, we define $k^{-}$ and $k^{+}$ as the number of $\lambda < 0.7$ and $\lambda \geq 0.7$ respectively. Therefore we have 
\begin{align*}
k^{+}\left(1-\frac{1}{e^{-0.7}}\right)&=\sum_{i: \lambda_i\ge 0.7} \left(1-\frac{1}{e^{-0.7}}\right)\\
& \le \sum_{i: \lambda_i \ge 0.7} 1-exp(-\lambda_i) \\
&\le \sum_{i} 1-exp(-\lambda_i)\\
& \stackrel{(a)}{\le} d(1+\epsilon)
\end{align*} 
where $(a)$ follows from Lemma~\ref{lem:dpoi1} and the fact $\dpoi=\sum_i 1-exp(-\lambda_i)$.
Hence we have $k^{+}\le \frac{d(1+\frac2n)}{1-\frac{1}{e^{0.7}}}\le 2d(1+\frac2n)$ and consequently  $k^{-}\ge k- \frac{d(1+\frac2n)}{1-\frac{1}{e^{0.7}}}$. For the type entropy we know 
\begin{align}
\sum_{\lambda_i \ge 0.7} \frac12 \log \left( 2 \pi e (\lambda_i+\frac{1}{12})\right) &\stackrel{(a)}{\le} \frac12 k^{+} \log \left( 2 \pi e (\frac{n^{+}}{k^{+}}+\frac{1}{12})\right)\nonumber\\
&\stackrel{(b)}{\le} d(1+\frac2n) \log \left(\pi e (\frac{n}{d(1+\frac2n)}+\frac{1}{6})\right)\nonumber\\
&= \left(d \log (\frac{n}{d}+\frac16)+d \log \pi e\right)(1 + o_d(1)) \label{lbound3}
\end{align}
where $(a)$ is by concavity of logarithm and $(b)$ is by monotonicity.
\begin{Lemma}
\label{lem:distinctlower}
\begin{align*}
\bar{R}(\pd^n) &\geq \left(\log {k\choose d}- d \log \left(\frac{n}{d}+\frac16\right)-d\log \pi e\right)(1+o_d(1)) -\sum_{\lambda_i < 0.7} (3\lambda_i -\lambda_i \log \lambda_i) .
\end{align*}
\end{Lemma}
\begin{proof}
~\eqref{lbound1},~\eqref{eq:entropy}, and~\eqref{lbound3} leads to the theorem. 
\end{proof}

\section{Expected redundancy of unordered power-law envelope}
\label{sec:zipf}
To use Lemmas~\ref{lem:distinctupper} and~\ref{lem:distinctlower} we need to bound the number of distinct elements that appear from any distribution in the envelope class $\uenv$ in addition to calculating the last summation in Lemma~\ref{lem:distinctlower}. For a distribution $p \in \uenv$ the number of distinct elements is 
\begin{align*}
\EE[\Dn] 
&= \sum^k_{i=1} \EE[\II_{\mu_i > 0}] \\
&= \sum^k_{i=1} 1 - (1-p_i)^n \\
&= \sum^k_{i=1} 1 - (1-p_{(i)})^n \\
&\leq \sum^k_{i=1} 1 - (1-f(i))^n \\
& \leq \sum_{i:f(i) \geq 1/n} 1 +  \sum_{i:f(i)< 1/n} 1 -(1-f(i))^n \\
& \leq \sum_{i:f(i) \geq 1/n} 1 +  \sum_{i:f(i) < 1/n} nf(i).
\end{align*}
Thus we need to bound the number of elements with envelope $\geq 1/n$ and the sum of envelopes for elements that are less than $1/n$.
For $\uialphaen$, the first term is $\leq (n/c)^{1/\alpha}$ and the second term is 
\begin{align*}
\leq \sum^k_{i =  (n/c)^{1/\alpha}} cn i^{-\alpha} 
& \leq \frac{c}{\alpha-1} n (n/c)^{\frac{1-\alpha}{\alpha}}\\
& \leq \frac{c^2}{\alpha-1} n^{1/\alpha}.
\end{align*}
Combining these, we get
\[
d \le \left(\frac{1}{c^{1/\alpha}} +  \frac{c^2}{\alpha-1}\right) n^{1/\alpha}.
\]
 
For $\uialpha$, we calculate $\sum_{\lambda_i < 0.7} \lambda_i$ and $\sum_{\lambda_i< 0.7} -\lambda_i \log \lambda_i$ for $\alpha > 1$.
In the below calculations ``$\approx$'' means that the quantities are equal up-to a multiplicative factor of $1+o_n(1)$.
\begin{align*}
n^{-}= \sum_{\lambda_i < 0.7} \lambda_i =& \sum_{i=\lfloor (\frac{10n}{7C_{k,\alpha}})^{\frac{1}{\alpha}}\rfloor+1}^k n\frac{i^{-\alpha}}{C_{k,\alpha}} \\
 \approx& n\int_{\left(\frac{10n}{7C_{k,\alpha}}\right)^{1/\alpha}}^k  \frac{i^{-\alpha}}{C_{k,\alpha}} di\\
 \approx& \frac{n}{(\alpha-1)C_{k,\alpha}} 
 \left(\left(\frac{10n}{7C_{k,\alpha}}\right)^{-(\alpha-1)/\alpha} - k^{-(\alpha-1)}\right)\\
 \approx& \frac{n}{(\alpha-1)C_{k,\alpha}} \left(\frac{10n}{7C_{k,\alpha}}\right)^{-(\alpha-1)/\alpha}\\
 =& \frac{7}{10(\alpha-1)}\left(\frac{10n}{7C_{k,\alpha}}\right)^{\frac{1}{\alpha}}
\end{align*}

Now we calculate the other summation. For $\uialpha$, 

\begin{align*}
 \sum_{\lambda_i < 0.7} -\lambda_i \log \lambda_i =& \sum_{i=\lfloor (\frac{10n}{7C_{k,\alpha}})^{\frac{1}{\alpha}}\rfloor+1}^k  n\frac{i^{-\alpha}}{C_{k,\alpha}} \log (\frac{C_{k,\alpha} i^{\alpha}}{n})\\
 =&  n^{-}\log(\frac{C_{k,\alpha}}{n}) + \frac{n\alpha}{C_{k,\alpha}}\sum_{i=\lfloor (\frac{10n}{7C_{k,\alpha}})^{\frac{1}{\alpha}}\rfloor+1}^k i^{-\alpha} \log i\\
 \le & n^{-}\log(\frac{C_{k,\alpha}}{n}) + \frac{n\alpha}{C_{k,\alpha}}\int_{\left(\frac{10n}{7C_{k,\alpha}}\right)^{1/\alpha}-1}^k i^{-\alpha} \log i di\\
 \le&  n^{-}\log(\frac{C_{k,\alpha}}{n}) + \frac{n\alpha}{C_{k,\alpha}}\left[\frac{x^{1-\alpha}((\alpha-1)\log x + 1)}{(\alpha-1)^2}\right]_{k}^{\left(\frac{2n}{C_{k,\alpha}}\right)^{1/\alpha}}+ \frac{n\alpha}{C_{k,\alpha}}\frac{1}{\alpha}\left(\frac{10n}{7C_{k,\alpha}}\right)^{-1} \log \left(\frac{10n}{7C_{k,\alpha}}\right)\\
 \le&n^{-}\log(\frac{C_{k,\alpha}}{n})+ \frac{n}{C_{k,\alpha}} \frac{1}{\alpha-1} (\frac{10n}{7C_{k,\alpha}})^{\frac{1-\alpha}{\alpha}} \log \frac{10n}{7C_{k,\alpha}} +\frac{n\alpha}{C_{k,\alpha}(\alpha-1)^2}\left(\frac{10n}{7C_{k,\alpha}}\right)^{\frac{1}{\alpha}-1}\\
 +&\frac{n}{C_{k,\alpha}}\left(\frac{10n}{7C_{k,\alpha}}\right)^{-1} \log \left(\frac{10n}{7C_{k,\alpha}}\right)\\
\le& \frac{7}{10(\alpha-1)}\left(\frac{10n}{7C_{k,\alpha}}\right)^{\frac{1}{\alpha}} \log (\frac{C_{k,\alpha}}{n})+\frac{7}{10(\alpha-1)} (\frac{10n}{7C_{k,\alpha}})^{\frac{1}{\alpha}} \log \frac{10n}{7C_{k,\alpha}} \\
+& \frac{7\alpha}{10(\alpha-1)^2}(\frac{10n}{7C_{k,\alpha}})^{\frac{1}{\alpha}}+\frac{7}{10} \log \frac{10n}{7C_{k,\alpha}}\\
=& \frac{11.2\alpha-4.2}{10(\alpha-1)^2}\left(\frac{10n}{7C_{k,\alpha}}\right)^{\frac{1}{\alpha}}+\frac{7}{10} \log \frac{10n}{7C_{k,\alpha}}
 \end{align*}

Substituting the above bounds in Lemmas~\ref{lem:distinctupper} and ~\ref{lem:distinctlower} results in the following theorem.
\begin{Theorem}
\label{thm:zipf}
For $k>n$, $c_1= \left(\frac{1}{c^{1/\alpha}} +  \frac{c^2}{\alpha-1}\right)$, $c'_1= \left(C_{k,\alpha}^{1/\alpha} +  \frac{C_{k,\alpha}^{-2}}{\alpha-1}\right)$, and $c_2=\frac{32.2\alpha-25.2}{10(\alpha-1)^2}\left(\frac{10}{7C_{k,\alpha}}\right)-c'_1 \cdot \log \pi e$
\begin{align*}
 \bar{R}(\uialphaen^n) &\ge \bar{R}(\uialpha^n) \ge \left(\log {k \choose c'_1 n^{\frac{1}{\alpha}}}- c'_1(1-\frac{1}{\alpha}) n^{\frac{1}{\alpha}} \log \frac{n}{c'_1}\right)(1+o_n(1)) -c_2 \cdot n^{\frac{1}{\alpha}}-\frac{7}{10}\log \frac{10n}{7C_{k,\alpha}},
\end{align*}
and 
\begin{align*}
\bar{R}(\uialphaen^n) &\le \log {k \choose c_1 n^{\frac{1}{\alpha}}}+ c_1(2-\frac{1}{\alpha}+2\log e) \cdot n^{\frac{1}{\alpha}} \log \frac{n}{c_1}+ \log (n+1).
\end{align*}
We can write both of the bounds above in order notation as 
\[
 \bar{R}(\uialphaen^n)= \Theta (n^{\frac{1}{\alpha}} \log k ).
\]
\end{Theorem}
\bibliographystyle{abbrvnat}
\bibliography{isit.bib}

\end{document}